\documentclass{article}

\usepackage[preprint]{neurips_2026}
\usepackage{amsmath}
\usepackage{graphicx}
\usepackage{amsthm}
\usepackage[utf8]{inputenc} 
\usepackage[T1]{fontenc}    
\usepackage{hyperref}       
\usepackage{url}            
\usepackage{booktabs}       
\usepackage{amsfonts}       
\usepackage{nicefrac}
\usepackage{thmtools}
\usepackage{microtype} 
\usepackage{multirow}
\usepackage{comment}
\usepackage{xcolor}         
\newtheorem{theorem}{Theorem}

\newtheorem{remark}{Remark}
\newtheorem{definition}{Definition}
\newtheorem{corollary}{Corollary}
\usepackage{subcaption}
\newcommand{\R}{\mathbb{R}}
\newcommand{\x}{\mathbf{x}}

\newcommand{\s}{\mathbf{s}}
\newcommand{\opt}{\mathbf{o}}
\newcommand{\sw}{\text{SW}}
\newcommand{\argmax}{\text{argmax}}
\def\RVCG{\mathrm{Rev}_{\text{VCG}}}
\def\Rev{\mathrm{Rev}_{\alpha}}

\title{A Unified Framework for Uniform-Price\\ Resource Allocation Mechanisms}

\author{%
  Ioannis Caragiannis \\
  Aarhus University
  \And Dimitris Fotakis \\
  National Technical University of Athens\\ and Archimedes, Athena RC
  \And Stratis Skoulakis\\
  Aarhus University
}

\begin{document}

\maketitle

\begin{abstract}
Mechanisms for allocating a divisible resource among strategic agents have been widely studied. The prominent paradigm is the proportional (Kelly) mechanism, which elicits a scalar bid per agent, allocates the resource proportionally, and charges payments equal to the bids. Follow-up mechanisms improve social welfare, but sacrifice simplicity by introducing complex allocation rules or unintuitive payments.

We introduce a unified framework for designing simple resource allocation mechanisms with proportional-style allocations and uniform pricing. Our framework yields a family of mechanisms that interpolate between the Kelly mechanism and the first-price auction. These mechanisms strictly improve upon Kelly’s efficiency guarantees, even achieving full efficiency in equilibrium, while also providing revenue guarantees relative to the VCG mechanism.
\end{abstract}

\section{Introduction}

The allocation of a divisible resource among strategic agents is a fundamental problem in algorithmic game theory, with numerous applications including packet routing~\cite{K97,S19,K98,La00}, congestion control~\cite{MB06,JT04,JT09,CV18}, scheduling, and blockchain systems~\cite{AN19,BGR24}. In its most basic form, a central authority seeks to allocate a perfectly divisible resource among $n$ strategic agents. Each agent $i \in [n]$ is associated with a concave valuation function $v_i:[0,1]\to\mathbb{R}_+$, where $v_i(x_i)$ denotes the value that agent $i$ derives from receiving an amount $x_i \in [0,1]$ of the resource. The objective of the central authority is to maximize the overall \emph{social welfare}, $
\sum_{i=1}^n v_i(x_i).$ 

Two fundamental challenges arise in this setting. First, each valuation function $v_i$ is private information known only to agent $i$, and agents may strategically misreport their preferences to improve their own utility. Second, valuation functions are infinite-dimensional objects, making them impossible to communicate efficiently to the central authority~\cite{JT04,JT09}.

The \emph{Kelly mechanism}~\cite{K97,K98}, also known as the \emph{proportional allocation mechanism}, provides an elegant solution to this problem. In the mechanism, each agent $i \in [n]$ submits a scalar bid $\theta_i \geq 0$ that is the amount of money agent $i$ pays. Then, the resource is allocated proportionally to the bids $x_i \propto \theta_i$. Due to its simplicity, the Kelly mechanism has been extensively studied (see Section~\ref{s:related}). In their seminal work, Johari and Tsitsiklis~\cite{JT04} showed that the Kelly mechanism admits a unique Nash equilibrium and that the social welfare at equilibrium is at least $75\%$ of the optimal social welfare.

\subsection{Our Contribution and Results}
In this work, we propose a general framework for designing \textit{simple and intuitive} mechanisms for resource allocation that are able to achieve social welfare arbitrarily close to the optimal social welfare.

More precisely, we propose a framework based on the idea of \textit{scalar-parameterized proxy valuation functions}. Each agent $ i \in [n]$ submits a scalar $\theta_i \geq 0$; then the mechanism considers, for each agent $i$, the \textit{proxy valuation function} $\tilde{v}_i(x_i) := \theta_i \cdot f(x_i)$,  where $f(x)$ is a predefined concave function. Our mechanism outputs the allocation $\x := \mathrm{argmax}_{\x} \sum_{i=1}^n \tilde{v}_i(x_i)$ maximizing the social welfare with respect to the proxy valuation function. Our framework then uses the Lagrange multiplier $\lambda \geq 0$ (of the respective mathematical program) as the \textit{unit price of the good} and charges $\lambda \cdot x_i$ to each agent $i$.

We apply the framework above in the case of proxy functions of the form,
\[\tilde{v}_i(x_i) := \theta_i \cdot \underbrace{\frac{x^{1-\alpha}}{1 - \alpha}}_{f(x)} ~~~~\text{for some parameter }\alpha \in (0,1].\]
In this case the mechanism above takes an \textit{$\alpha$-proportional form} where 
$x_i \propto \theta_i^{1/\alpha}$ and $
\lambda := \left(\sum_{i=1}^n \theta_i^{1/\alpha}\right)^\alpha$. As a result, for $\alpha = 1$ our framework recovers the famous Kelly mechanism while for $\alpha \rightarrow 0$ our framework recovers the First-Price Auctions. We provide the following results for the $\alpha$-proportional mechanism. 

\textit{Existence and Uniqueness of Nash Equilibrium.} We establish that for any $\alpha \in (0,1]$ the $\alpha$-proportional mechanism always admits a Nash equilibrium. Moreover, under the mild assumption that each valuation function is increasing and differentiable, we additionally show that the Nash equilibrium is unique.

\textit{Improved Social Welfare.} The main contribution of our work is establishing that the social welfare of the Nash Equilibrium of the $\alpha$-proportional mechanism 
is at least $\frac{1+2\sqrt{\alpha}}{(1+\sqrt{\alpha})^2}$ times the optimal social welfare,
%
\[\mathrm{Social~Welfare~at~NE~} \geq \frac{1+2\sqrt{\alpha}}{(1+\sqrt{\alpha})^2}\cdot \mathrm{Optimal~Social~Welfare}. \]
Our result recovers the $4/3$ Price of Anarchy (PoA) bound of Johari and Tsitsiklis~\cite{JT04} for the Kelly mechanism as the special case of $\alpha = 1$. Most importantly, it shows that the $\alpha$-proportional mechanism admits PoA arbitrarily close to $1$ as $\alpha$ approaches $0$. We additionally provide a construction showing that our analysis is tight.

The latter is very surprising, in view of an unexpected sharp phase transition. As $\alpha$ approaches $0$, the $\alpha$-proportional mechanism approximates the \textit{First-Price Auction} where the bidders with maximum bid $\max_{i \in [n]} \theta_i$ share the good equally. Despite the fact that the Price of Anarchy of the $\alpha$-proportional mechanism approaches $1$ as $\alpha$ approaches $0$, we show that the PoA of the First-Price Auction is $n/2$ in our setting!

\textit{Revenue Guarantees.} We provide several theoretical results comparing the revenue produced by the $\alpha$-proportional mechanism with the revenue produced by the well-known VCG mechanism. In particular, we establish that at Nash Equilibrium the $\alpha$-proportional mechanism produces revenue that is at least $1/(1+\alpha)$ of the revenue produced by the VCG mechanism in the case of linear valuation functions or for $n =2$ in the case of general concave valuations. We also establish that, if all agents admit the exact same concave valuation function, the $\alpha$-proportional mechanism achieves at least $(n-1)/(n-1 + \alpha)$ fraction of the VCG revenue. 

\textit{Experimental Evaluations.} We experimentally evaluate the $\alpha$-proportional mechanism in case the agents repeatedly select their bids via online learning algorithms. Our experimental evaluation indicates that online learning dynamics converge to the unique Nash Equilibrium. We also experimentally evaluate the resulting social welfare for various values of $\alpha \in (0,1]$. Our experimental evaluation verifies the fact that lower values of $\alpha$ lead to higher social welfare.

\subsection{Related Work}\label{s:related}
The problem of designing mechanisms for resource allocation has been extensively studied over the years. As already discussed, in his seminal work, Kelly~\cite{K97} (see also~\cite{K98}) proposed the proportional mechanism, while later Johari and Tsitsiklis~\cite{JT04} established that every Nash equilibrium achieves at least $3/4$ of the optimal social welfare. More general equilibrium concepts and  efficiency objectives different than the social welfare have been considered by Caragiannis and Voudouris~\cite{CVa,CV18}, Christodoulou et al.~\cite{CS15}, Syrgkanis and Tardos~\cite{ST13}, and Correa et al.~\cite{CSM13}. The allocation rule resulting from our $\alpha$-proportional framework has been studied in the economics literature in the context of Tullock contests~\cite{T75,T82} (see also~\cite{GNR25,V16}). However, Tullock contests operate under a `pay-your-bid' pricing rule. A key distinguishing feature of our framework is that it pairs this allocation rule with a uniform pricing scheme, which allows it to achieve arbitrarily high efficiency without sacrificing the fairness of a fixed unit price.

There have been several attempts to improve upon the Kelly mechanism. Sanghavi and Hajek~\cite{SH04} proposed a mechanism with a Price of Anarchy better than $8/7$ in the case of two agents. Johari and Tsitsiklis~\cite{JT09} introduced the idea of scalar-parameterized valuation mechanisms. In this framework, each agent $i \in [n]$ reports a scalar $\theta_i \geq 0$ that defines a proxy valuation function of the form $\theta_i \cdot f(x_i)$. Johari and Tsitsiklis~\cite{JT09} then compute the allocation that maximizes social welfare with respect to these proxy functions and apply VCG payments. They showed that this mechanism admits a Nash equilibrium that induces a revenue-maximizing allocation. However, on the negative side, the PoA of this approach can be arbitrarily large; that is, there exist Nash equilibria with arbitrarily low social welfare. The same mechanism was independently proposed in~\cite{YH22b}. Maheswaran and Basar~\cite{MB06} proposed the ESPA mechanism, which uses the proportional allocation rule $x_i \propto \theta_i$, but instead of charging each agent $\theta_i$, agent $i$ pays
\[
\left( \sum_{j \neq i} \theta_j \right)
\cdot
\int_0^{\theta_i}
g\!\left(t+\sum_{j\neq i}\theta_j\right)
\cdot
\left(t+\sum_{j\neq i}\theta_j\right)^{-2}
\, dt,
\]
where $g(\cdot)$ is any increasing function. Maheswaran and Basar~\cite{MB06} showed that ESPA achieves a PoA of $1$ and admits a unique Nash equilibrium. However, a crucial drawback of ESPA is the complexity of its payment rule (see Appendix~\ref{app:related}).

In this regard, we emphasize an important advantage of our $\alpha$-proportional mechanism compared to the mechanisms of Johari and Tsitsiklis~\cite{JT09} and Maheswaran and Basar~\cite{MB06}. Beyond its simplicity, the $\alpha$-proportional mechanism allocates resources proportionally to $\theta_i^{1/\alpha}$ and admits a \emph{fixed unit price} $\lambda$ for the good. In particular, each agent pays $\lambda \cdot x_i$. In contrast, under both the mechanisms of Johari and Tsitsiklis~\cite{JT09} and Maheswaran and Basar~\cite{MB06}, different agents effectively pay different unit prices.

\section{Preliminaries} 
In the classical \textit{resource allocation} setting, a central authority wants to allocate a divisible good of unit quantity to $n$ strategic agents. Each agent $i \in [n]$ admits a valuation function $v_i(x_i)$ encoding each value for an $x_i$ share of the good. We assume that each valuation function $v_i(x_i)$ is concave and normalized $v_i(0) = 0$.

The central authority wants to find an allocation $\x \in \Delta_n := \{\x \in \mathbb{R}^n ~:~\sum_{i=1}^n x_i = 1 ~~\text{and}~~x_i \geq 0\}$
that maximizes the overall sum of the valuations of the agents, 
also known as \textit{social welfare},
\[\sw(\x) := \sum_{i=1}^n v_i(x_i).\]
We denote by $\opt^\star := \mathrm{argmax}_{\x \in \Delta_n}\sum_{i=1}^n v_i(x_i)$ the allocation that maximizes social welfare.

\subsection{The VCG Mechanism}
Each valuation function $v_i(x_i)$ is \textit{private information} of the agent $i \in [n]$. The challenge is that agents are selfish and strategic and thus may misreport their valuation function $\tilde{v}_i(x_i)$ in order to get a higher fraction of the good.

In their seminal work, Vickrey, Clarke, and Groves (see~\cite{Milgrom_2004}) provide a general mechanism incentivizing agents to \textit{truthfully} report their valuation functions. In particular,
the VCG mechanism asks each agent $i \in [n]$ to \textit{report} a valuation function $\tilde{v}_i : [0,1] \mapsto \mathbb{R}$. We remark that each agent can report $\tilde{v}_i(x_i) \neq v_i(x_i)$. VCG then computes the allocation $\tilde{\opt} \in \Delta_n$ maximizing the social welfare with respect to the reported functions,
\[\tilde{\opt} := \mathrm{argmax}_{\x \in \Delta_n}\sum_{i=1}^n \tilde{v}_i(x_i).\]
and charges each agent $i \in [n]$,
\[p_i :=\sum_{ j \neq i} \tilde{v}_j(\tilde{o}^{-i}_j) - \sum_{ j \neq i} \tilde{v}_j(\tilde{o}_j)\]
where $\tilde{\opt}^{-i} := \mathrm{argmax}_{\x}\sum_{j \neq i}^n \tilde{v}_j(x_j)$. Thus, the utility of agent $i \in [n]$ is $u_i(\tilde{v}_i,\tilde{v}_{-i}) := v_i(\tilde{\opt}_i) - p_i.$ 

VCG is a \textit{Dominant Strategy Incentive Compatible} mechanism, which means that no matter the declared valuations $\{\tilde{v}_{j}\}_{j \neq i}$, agent $i \in [n]$ always maximizes their utility by selecting $\tilde{v}_i = v_i$. In other words, VCG incentivizes each agent to \textit{truthfully} report their valuation function, meaning that the resulting allocation $\tilde{\opt}$ coincides with the true social welfare maximizer $\opt^\star$.
%
As a result, despite the strategic nature of the agents, VCG is able to compute the allocation that maximizes the social welfare!

\subsection{The Kelly Mechanism}
Despite the fact that the VCG mechanism is able to compute an allocation that maximizes the social welfare, it admits an evident strong caveat:  \textit{ Each agent $i \in [n]$ needs to report an infinite number of values in order to describe its valuation function $v_i(x_i)$.}

The latter renders VCG completely impractical for resource allocation. In response to the latter problem, the \textit{Kelly mechanism} (or \textit{proportional mechanism}) has been proposed~\cite{K97}. In the Kelly mechanism, each agent $i \in [n]$ submits a \textit{scalar bid} $\theta_i \geq 0$. Then the mechanism assigns each agent $i \in [n]$ a fraction $x_i \propto \theta_i$ proportional to the submitted bids and charges it $\theta_i$. In particular, given a set of bids $\mathbf{\theta}:= (\theta_1,\ldots,\theta_n)$ the utility of agent $i \in [n]$ equals
\[u_i(\theta_i,\theta_{-i}) := v_i\left(\frac{\theta_i}{\sum_{j=1}^n \theta_j} \right)- \theta_i.\]
In the Kelly mechanism, each agent $i \in [n]$ needs to select its bid $\theta_i \geq 0$ in order to maximize its individual utility. A Nash Equilibrium is a stable state where no agent has incentive to change its bid.
\begin{definition}
A set of bids $\theta^\star := (\theta^\star_1,\ldots,\theta^\star_n)$ is a Nash Equilibrium if and only if for each agent $i \in [n]$, $u_i(\theta^\star_i,\theta^\star_{-i}) \geq u_i(\theta_i,\theta^\star_{-i})$ for all $\theta_i \geq 0$.     
\end{definition}
In their seminal work, Johari and Tsitsiklis~\cite{JT04} showed that any Nash Equilibrium of the Kelly mechanism admits social welfare that is at least $3/4$ of the optimal social welfare.

\begin{theorem}\cite{JT04}
Let an equilibrium $\theta^\star := (\theta^\star_1,\ldots,\theta^\star_n)$ of the Kelly mechanism and $\x(\theta^\star)$ its induced allocation, $x_i(\theta^\star) := \theta^\star_i/\sum_{j=1}^n \theta^\star_j$. Then the Price of Anarchy (PoA) defined as,
\[\mathrm{PoA}:= \frac{\sw(\opt^\star)}{\sw(\x(\theta^\star))} \leq \frac{4}{3}.\]
\end{theorem}

\textbf{VCG vs Kelly} On the positive side, the Kelly mechanism is way simpler, more intuitive and requires significantly less information exchange than VCG. On the negative side, the Kelly mechanism at its equilibrium achieves only $3/4$ of the optimal social welfare.

\section{Scalar-Parametrized Proxy Functions via Dual Pricing}\label{s:framework}

We introduce a framework for designing \textit{resource allocation} mechanisms that incorporate the simplicity and practicality of the Kelly mechanism together with the higher efficiency of VCG. Our framework builds on the idea of \textit{scalar-parametrized proxy functions} that have also been considered in \cite{JT09}. The idea is that each agent $i \in [n]$ reports a positive scalar $\theta_i \geq 0$ and, then, the mechanism considers as valuation of the agent the function $\tilde{v}_i(x_i) = \theta_i \cdot f(x_i)$ where $f(x)$ is a predefined concave function e.g. $f(x) := \sqrt{x}$ or $f(x):= \ln x$.

\begin{remark}
Our framework combines VCG and the Kelly mechanism, in the sense that each agent $ i\in [n]$ declares a whole proxy valuation function $\tilde{v}_i(x_i) := \theta_i \cdot f(x_i)$ by reporting a single parameter $\theta_i \geq 0$. 
\end{remark}

We consider as predefined function $f(x)$ of the mechanism, $f(x) := x^{1 - \alpha}/(1-\alpha)$ for some parameter $\alpha \in (0,1]$. Given the scalars $\mathbf{\theta}:=(\theta_1,\ldots,\theta_n)$ selected by the agents, our framework computes the allocation $\mathbf{x}:= (x_1,\ldots, x_n)$ maximizing the social welfare with respect to the proxy valuation functions $\{\tilde{v}_i := \theta_i \cdot x_i^{1-\alpha}/(1-\alpha)\}_{i=1}^n$. The latter can be done via solving the following primal/dual pair of programs
\[
\begin{aligned}
& \textbf{Primal}~~~~~~~~~~~~~~~~~~~~~~~~~~~~~~~~~~~~~~~  
&& \textbf{Dual} \\
& \max_{x_i \ge 0} \;\sum_{i=1}^n \theta_i \frac{x_i^{1-\alpha}}{1-\alpha}~~~~~~~~~~~~~~~~~~~~~~~~~~~~~~~~~~~~ 
&& \min_{\lambda \ge 0} \;\lambda + \frac{\alpha}{1-\alpha}\,\lambda^{\frac{\alpha-1}{\alpha}} \sum_{i=1}^n \theta_i^{1/\alpha} \\
& \text{s.t. } \sum_{i=1}^n x_i \le 1
~~~~~~~~~~~~~~
&&
\end{aligned}
\]
Our framework assigns to each agent $i \in [n]$ a quantity $x_i \geq 0$ of the good with respect to the optimal solution $\x := (x_1,\ldots,x_n)$ of the primal program. The dual variable $\lambda \geq 0$ determines the \textit{unit price of the good}. More precisely, each agent $i \in [n]$ pays $\lambda \cdot x_i$ for the $x_i$ share of the good that the agent got. 

To this end, one may wonder why using the dual variable $\lambda \geq 0$ as the unit price of the good. The reason is that it provides the so-called \textit{utility maximization property} with respect to $\tilde{v}_i(x_i)$.

\begin{restatable}{lemma}{lemmaone}
\label{l:1}
Let $(\x,\lambda)$ be the solution of the primal-dual program above. Then for each agent $i \in [n]$,
\[x_i := \mathrm{argmax}_{x_i \geq 0}~\left\{ \theta_i \cdot \frac{x_i^{1 - \alpha}}{1 - \alpha} - \lambda \cdot x_i\right\}.\]
\end{restatable}

In simpler terms, based on the declared proxy functions $\{\hat{v}_i\}_{i=1}^n$, our mechanism computes a unit-price of the good $\lambda \geq 0$ and an allocation $x_i$ for each agent $i \in [n]$ such that $x_i$ maximizes the agent's utility under the assumption that $\tilde{v}_i(x_i) := \theta_i \cdot x_i^{1-\alpha}/(1-\alpha)$ is the actual valuation of the agent. The proof of Lemma~\ref{l:1} follows by KKT conditions and is presented in Appendix~\ref{app:framework}.  
 
Due to the specific structure of the valuation functions $\tilde{v}_i(x_i) := \theta_i \cdot x^{1 - \alpha}/(1-\alpha)$ in Lemma~\ref{l:2}, we establish that the optimal solution of the primal/dual program above admits the following form.

\begin{restatable}{lemma}{lemmatwo}
\label{l:2}
The optimal solution of the primal/dual problem above is $x_i := \theta_i^{1/\alpha}/(\sum_{j\in [n]} \theta_j^{1/\alpha})$. Moreover, the unit price $\lambda:= \left(\sum_{j\in [n]} \theta_j^{1/\alpha}\right)^\alpha$.
\end{restatable}

As a result, given a set of bids $\mathbf{\theta} := (\theta_1,\ldots,\theta_n)$, the utility of each agent $i \in [n]$ equals,
\[u_i(\theta_i,\theta_{-i}):= \underbrace{v_i\left(\frac{\theta^{1/\alpha}_i}{\sum_{j = 1}^n \theta_j^{1/\alpha}}\right)}_{\text{valuation of agent } i \in [n]} ~~~-~~~ 
\underbrace{\frac{\theta^{1/\alpha}_i}{\left(\sum_{j=1}^n \theta_j^{1/\alpha} \right)^{1 - \alpha}}}_{\text{payment of agent }i \in [n]}.\]
As a result, our framework comes as a generalization of the Kelly mechanism for various values of $\alpha \in (0,1]$ where given a set of bids $\theta := (\theta_1,\ldots,\theta_n)$ the resource is allocated proportionally to $\theta_i^{1/\alpha}$. Due to this reason, we also call our mechanism as \textit{$\alpha$-proportional mechanism}.

\textit{Connection to Kelly Mechanism.} Our $\alpha$-proportional framework captures the seminal Kelly mechanism for resource allocation. In the Kelly mechanism, each agent $i \in [n]$ submits a bid $\theta_i \geq 0$. Each agent $i \in [n]$ is then allocated a fraction $x_i \propto \theta_i$ and pays $\theta_i$.   

\textit{Connection to First-Price Auctions.} As $\alpha \rightarrow 0$ the $\alpha$-proportional mechanism approaches the First-Price Auction where the bidders with the highest bids share in equal shares the good.

\subsection{Paper Organization and Results}
In Section~\ref{s:NE} we show that the $\alpha$-proportional mechanism always admits a Nash Equilibrium while we show that under very mild assumptions on the valuations functions, the Nash Equilibrium is additionally unique. In Section~\ref{s:PoA} we present the main result of our work establishing that the Price of Anarchy of the $\alpha$-\textit{proportional mechanism} is upper bounded by 
\[ \mathrm{PoA} \leq \frac{(1+\sqrt{\alpha})^2}{1+2\sqrt{\alpha}}.\]
This means that $\alpha \rightarrow 1$, our framework can achieve social welfare arbitrarily close to the optimal social welfare. The latter improves on the $3/4$ approximation of the Kelly mechanism. In Section~\ref{s:rev} we provide revenue guarantees of the $\alpha$-proportional mechanism with respect to the revenue produced by the VCG. Finally, in Section~\ref{s:exp}, we experimentally evaluate our framework with respect to classical online learning algorithms.

\section{Existence and Uniqueness of Nash Equilibrium}\label{s:NE}
In this section, we establish that for any $\alpha \in (0,1]$, the $\alpha$-proportional mechanism always admits a Nash Equilibrium. We additionally show that if the valuation functions $v_i(x_i)$ are increasing and differentiable, the Nash Equilibrium is unique. 

To simplify notation, we consider the change of variables $s_i := \theta_i^{1/\alpha}$ and thus the utility function of agent $i \in [n]$, takes the following form,
\[u_i(s_i,\s_{-i}) := v_i\left(\frac{s_i}{s_i + \sum_{j \neq i} s_j} \right) - \frac{s_i}{\left(s_i + \sum_{j \neq i} s_j \right)^{1-\alpha}}.\]
The \textit{valuation function} $v_i(x_i)$ is concave with respect to $x_i$, however the \text{utility function} $u_i(s_i,s_{i})$ is not necessarily concave with respect to $s_i$ and it is well-known that games with non-concave utility functions may not admit a Nash Equilibrium~\cite{DSZ21}. However, by leveraging the structure of the utility functions $\{u_i(s_i,s_{-i})\}_{i=1}^n$, we establish the existence of a Nash Equilibrium via associating it with the maximizer of the concave function defined up next.

\begin{definition}\label{d:1}
The superdifferential $\partial v(x)$ of a concave function $v : [0,1] \to \mathbb{R}$ is defined as
\[
\partial v(x) := \{ g \in \mathbb{R} : v(y) \le v(x) + g(y-x), \ \forall y \in [0,1]\}.
\]
We also denote with $v'_{+}(x):= \max\{g \in \partial v(x) \}$.
\end{definition}

Given the concave valuation function $v_1(x_1),\ldots,v_n(x_n)$ of the agents, we consider the following potential function
\[
\Psi(x) := \sum_{i=1}^n \int_0^x v'_{i+}(z) \cdot \frac{1 - z}{1 - (1-\alpha)z}dz.\] 

In Theorem~\ref{t:main} we establish the fact that the allocation $\x \in \Delta_n$ maximizing $\Psi(\x)$ corresponds to a Nash Equilibrium and vice versa. The proof of Theorem~\ref{t:main} is presented in Appendix~\ref{app:nash}.

\begin{restatable}{theorem}{lemmathree}
\label{t:main}
Let $\x^\star \in \Delta_n$ be the allocation maximizing $\Psi(\x)$. Then the following hold,
\begin{enumerate}
    \item There exists $\mu > 0$ such that for any agent $i \in [n]$ with $x_i^\star > 0$, \[ \mu = v'_{i+}(x^\star_i) \cdot \frac{1 - x^\star_i}{1 - (1-\alpha)x_i^\star}.\]
\item The strategy profile $\s^\star \in \mathbb{R}_{+}^n$ defined as $s^\star_i = x_i^\star \cdot \mu^{1/\alpha}$ is a Nash Equilibrium of the $\alpha$-proportional mechanism. 
\end{enumerate}
\end{restatable}

Theorem~\ref{t:main} establishes the existence of Nash Equilibrium since $\Delta_n$ is a compact set and thus there is always a maximizer of $\Psi(\x)$. 

\begin{restatable}{theorem}{lemmafour}
\label{t:unique}
Let $\s^\star:= (s^\star_1,\ldots,s^\star_n)$ be a Nash Equilibrium of the $\alpha$-proportional mechanism.  In case where each $v_i(x_i)$ is differentiable then there exists $\mu > 0$ such that for each agent $ i \in [n]$,
\[v'_i(x_i) = \mu ~~~\text{ if } x^\star_i > 0 ~~~~\text{ and }~~~~v'_i(0) \leq \mu \text{ if } x^\star_i = 0.\]
The latter implies that $\x^\star$ maximizes $\Psi(\x)$. 
\end{restatable}

Theorem~\ref{t:unique} establishes the uniqueness of Nash Equilibrium if each valuation function $v_i(x_i)$ is differentiable and strictly increasing. In such a case, the function $\Psi(\x)$ is strictly concave and thus admits a unique maximizer in $\Delta_n$. Hence there exists a unique Nash Equilibrium.

\section{Bounding the Price of Anarchy}\label{s:PoA}
In this section, we provide a tight bound of the Price of Anarchy on the $\alpha$-proportional mechanism. Our result is formally stated and proven in Theorem~\ref{t:PoA}.
\begin{theorem}\label{t:PoA}
The price of anarchy of the $\alpha$-proportional mechanism is at most $\frac{(1+\sqrt{\alpha})^2}{1+2\sqrt{\alpha}}$.
\end{theorem}

\begin{proof}
Consider an allocation instance with $n$ players, in which player $i\in [n]$ has the monotone non-decreasing valuation function $v_i: [0,1]\rightarrow \R_{\geq 0}$. Let $\theta=(\theta_1,\ldots,\theta_n)$ be a Nash Equilibrium of the $\alpha$-proportional mechanism. For convenience, we will use the equivalent transformation $s_i=\theta_i^{1/\alpha}$. Hence, the utility of player $i$ at equilibrium is
    \begin{align*}
        u_i(\s)=v_i\left(\frac{s_i}{\sum_{j\in [n]}{s_j}}\right)-\frac{s_i}{\left(\sum_{j\in [n]}{s_j}\right)^{1-\alpha}}. 
    \end{align*}
Since $\s$ is a Nash Equilibrium, then $s_i$ maximizes the function $u_i(\cdot, \s_{-i})$. Thus, there exists a supergradient $g_i \in \partial v_i(x_i)$ such that
    \begin{align*}
        g_i\cdot \frac{\sum_{j\not=i}{s_j}}{\left(\sum_{j\in [n]}{s_j}\right)^2}-\frac{\sum_{j\in [n]}{s_j}-(1-\alpha)s_i}{\left(\sum_{j\in [n]}{s_j}\right)^{2-\alpha}} &=0.
    \end{align*}
    Rearranging and using the definition $x_i = s_i / \sum_{j\in[n]} s_j$, we get
    \begin{align}\label{eq:g_i}
    g_i &= \frac{1-(1-\alpha)x_i}{1-x_i}\cdot \left(\sum_{j\in [n]}{s_j}\right)^{\alpha}
    \end{align}
    
    Using the definition of the supergradient $g_i$ and the fact that $v_i(0)=0$, we have $v_i(0)\leq v_i(x_i)+g_i \cdot (0-x_i)$ which implies that $v_i(x_i)\geq g_i\cdot x_i$. Hence, 
    \begin{align}\label{eq:sw}
    \sw(\x) &= \sum_{i\in [n]}{v_i(x_i)} \geq \sum_{i\in [n]}{g_i\cdot x_i}.
    \end{align}    
    Let $\opt^\star=(o^\star_1, \ldots, o^\star_n)$ denote the social-welfare optimal allocation. Let $m:=\argmax_{i\in [n]}{g_i}$ and using the definition of the supergradients $g_i$, we have
    \begin{align}\nonumber
        \sw(\opt^\star) &= \sum_{i\in [n]}{v_i(o^\star_i)} \leq \sum_{i\in [n]}{\left(v_i(x_i)+g_i\cdot (o^\star_i-x_i)\right)} = \sw(\x)+\sum_{i\in [n]}{g_i\cdot o^\star_i}-\sum_{i\in [n]}{g_i \cdot x_i}\\\label{eq:opt}
        &\leq \sw(\x)+g_m\cdot \sum_{i\in [n]}{o^\star_i}-\sum_{i\in [n]}{g_i\cdot x_i}=\sw(\x)+g_m-\sum_{i\in [n]}{g_i\cdot x_i}.
    \end{align}
    Using Equations (\ref{eq:sw}) and (\ref{eq:opt}), we have that the price of anarchy is
    \begin{align}\label{eq:sw_star_over_sw}
    \text{PoA} &\leq \frac{\sw(\opt^\star)}{\sw(\x)} \leq \frac{\sw(\x)+g_m-\sum_{i\in [n]}{g_i\cdot x_i}}{\sw(\x)} \leq \frac{g_m}{\sum_{i\in [n]}{g_i\cdot x_i}}.
    \end{align}
    The third inequality follows by subtracting the non-negative quantity $\sw(\opt^\star)-\sum_{j\in [n]}{g_j\cdot x_j}$
    from both the numerator and the denominator. Now, notice that Equation (\ref{eq:g_i}) implies that $g_i\cdot x_i \geq x_i\cdot \left(\sum_{j\in [n]}{s_j}\right)^{\alpha}$ and, hence, $\sum_{i\not=m}{g_i\cdot x_i}\geq (1-x_m)\cdot \left(\sum_{j\in [n]}{s_j}\right)^{\alpha}$. Using this observation and the definition of $g_i$ from Equation (\ref{eq:g_i}), Equation (\ref{eq:sw_star_over_sw}) yields
    \begin{align*}
        \text{PoA} &\leq \frac{g_m}{g_m\cdot x_m+\sum_{i\not=m}{g_i\cdot x_i}} \leq \frac{1-(1-\alpha)x_m}{\left(1-(1-\alpha)x_m\right)x_m+(1-x_m)^2}=\frac{1-(1-\alpha)x_m}{1-x_m+\alpha x_m^2}.
    \end{align*}
    The RHS is maximized for $x_m=\frac{1}{1+\sqrt{\alpha}}$ to the desired value of $\frac{(1+\sqrt{\alpha})^2}{1+2\sqrt{\alpha}}$.
\end{proof}

In Theorem~\ref{t:lower_bound}, we can show that the price of anarchy bound of Theorem~\ref{t:PoA} is tight. The proof is based on providing a specific instance of linear valuation functions, and showing that the ratio of the optimal social welfare and the social welfare of the Nash Equilibrium is exactly $\frac{(1+\sqrt{\alpha})^2}{1+2\sqrt{\alpha}}$.

\begin{restatable}{theorem}{theoremfive}
\label{t:lower_bound}
For $\alpha \in (0, 1]$, the price of anarchy the $\alpha$-proportional mechanism is at least $\frac{(1+\sqrt{\alpha})^2}{1+2\sqrt{\alpha}}$.
\end{restatable}

\subsection{A Sharp Phase Transition for First-Price Auctions}
Notice that as $\alpha \rightarrow 0$, the $\alpha$-proportional mechanism approaches the First-Price Auction~(FPA), where the highest bidders share the good equally. Interestingly, there exists a sharp phase transition between the case $\alpha \rightarrow 0$ and the FPA itself. We next show that the PoA for the FPA can be as large as $n / 2$ where $n$ is the number of agents.

Consider the case of $n$ agents with identical valuation functions, where each agent $i \in [n]$ has the valuation function
\[
v_i(x_i) =
\begin{cases}
n \cdot x_i & \text{if } x_i \leq 1/n, \\
1 & \text{otherwise}.
\end{cases}
\]

Notice that the optimal social welfare in this case is exactly $n$, achieved when each agent receives $1/n$ of the good. Now consider the strategy profile where $\theta_1 = \theta_2 = 1$ and $\theta_i = 0$ for all $i \geq 3$. This constitutes a Nash equilibrium. However, under this equilibrium, only agents $1$ and $2$ receive positive allocations, each obtaining $1/2$ of the good. Thus, the resulting social welfare is only $2$.

\section{Revenue Guarantees of the Proportional Mechanism}\label{s:rev}
	
In this section, we provide formal guarantees on the revenue of the $\alpha$-proportional mechanism with respect to VCG. The overall revenue of VCG, denoted as $\RVCG$ equals,
\[\RVCG := \sum_{i=1}^n \left( \sum_{ j \neq i} v_j(o^{-i}_j) - \sum_{ j \neq i} v_j(o^\star_j) \right).\]
where $\opt^{-i} :=  \mathrm{argmax}_{\x}\sum_{j \neq i}^n v_j(x_j)$.
At the same time, the overall revenue of the $\alpha$-proportional mechanism at a Nash Equilibrium $\theta^\star := (\theta^\star_1,\ldots,\theta^\star_n)$, denoted as $\mathrm{Rev}_{\alpha}$, equals
\[\mathrm{Rev}_{\alpha} := \lambda^\star = \left( \sum_{i=1}^n \theta_i^{\star 1/\alpha} \right)^{\alpha}.\]

In Theorem~\ref{t:linear} we establish that in case of linear valuation functions, the revenue of $\alpha$-proportional is at least $1/(1+\alpha)$ the revenue of VCG.

\begin{restatable}{theorem}{theoremsix}
\label{t:linear}
Let each agent $i \in [n]$ have a linear valuation function $v_i(x) = \zeta_i x$, with $\zeta_i \geq 0$. Then, $\RVCG \leq (1+\alpha)\Rev$. 
\end{restatable}

In the case of general concave valuations we establish that the revenue guarantees of the $\alpha$-proportional with respect to VCG for the special case of $n = 2$ and the case of identical valuation functions.

\begin{restatable}{theorem}{theoremseven}
\label{t:n-players}
For $n \geq 2$ agents with identical concave valuations, 
	\[ \RVCG \leq \Rev \left(1+\frac{\alpha}{n-1}\right) \]
\end{restatable}

\begin{restatable}{theorem}{theoremeight}
\label{t:two-players}
For $n=2$ agents with general concave valuations, $\RVCG \leq (1+\alpha)\Rev$\,. 
\end{restatable}

Our results indicate that the $\alpha$-proportional mechanism produces $1/(1+\alpha)$ of the revenue produced by the VCG.

\section{Experimental Evaluations}\label{s:exp}
In this section, we experimentally evaluate the $\alpha$-proportional mechanism under classical no-regret algorithms. We consider each agent $i \in [n]$ selecting each strategy $\theta^t_i \geq 0$ at each round $t \geq 1$ via running the Hedge algorithm. We experimentally evaluate both the convergence properties of such no-regret dynamics to the unique Nash Equilibrium of the game as well as the produced social welfare.

In particular, we consider valuation functions of the form $v_i(x_i) := \zeta_i \cdot x_i^\beta$ where $\zeta_i \in \mathrm{Unif}([0,1])$. We conduct the experiment for $n =5$ and all values $\beta:= \{0.5 , 0.75, 1\}$. We discretize the action space of each agent to be $[0,1]$ with granularity $k = 0.01$ and plot the exploitability $\mathrm{exp}(t) := \max_{i \in [n]} \left(\max_{\theta_i}U_i(\theta_i,\theta^t_{-i}) - U_i(\theta^t_i,\theta^t_{-i}) \right)$. Our experimental evaluations presented in Figure~\ref{f:1} suggest that as $T$ increases, $\mathrm{exp}(t)$ approaches $0$.

\begin{figure}[t]
    \centering

    \begin{minipage}{0.32\textwidth}
        \centering
        \includegraphics[width=\linewidth]{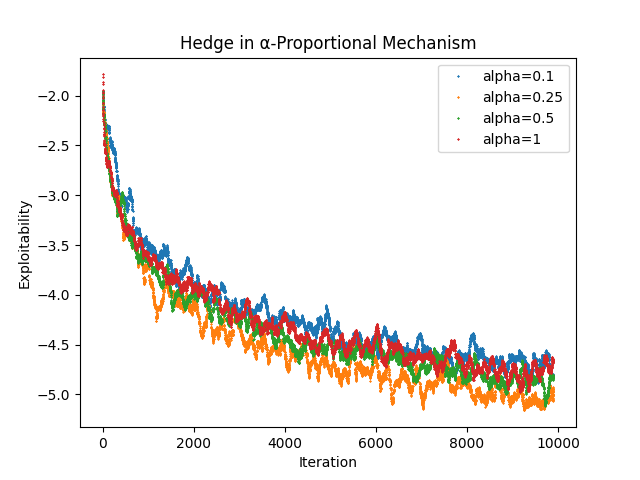}
        
        $\beta = 1$
    \end{minipage}
    \hfill
    \begin{minipage}{0.32\textwidth}
        \centering
        \includegraphics[width=\linewidth]{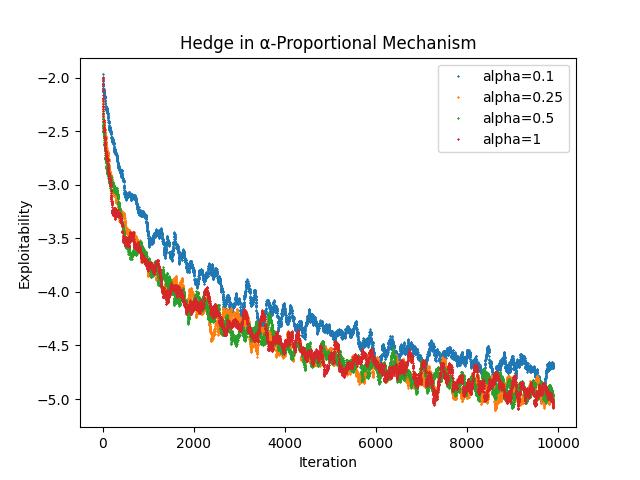}
        
        $\beta = 0.75$
    \end{minipage}
    \hfill
    \begin{minipage}{0.32\textwidth}
        \centering
        \includegraphics[width=\linewidth]{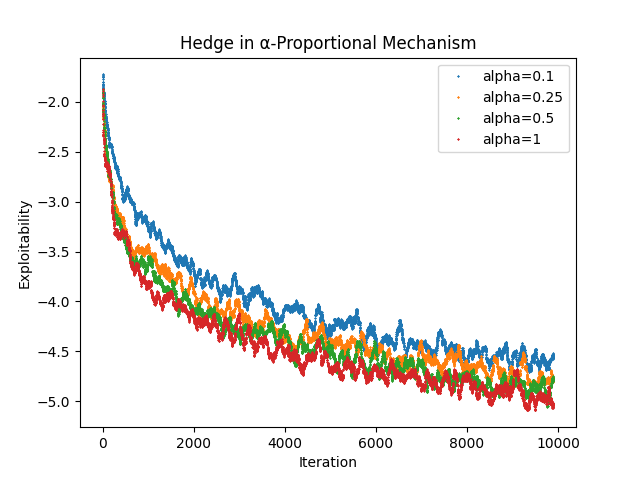}
        
        $\beta = 0.5$
    \end{minipage}

    \caption{Exploitability under Hedge dynamics for different values of $\beta$.}
    \label{f:1}
\end{figure}

Next, we evaluate the average social welfare produced by the online learning dynamics for various values of $\alpha \in (0,1]$ and number of agents $n = \{2,3,5,10\}$. In Table~$2$ we consider the case where $v_i(x_i) := \zeta_i \cdot x_i^\beta$ where $\zeta_i \in \mathrm{Unif}([0,1])$. In such cases where the valuation functions come from the same distribution, our experimental evaluations suggest that all values of $\alpha \in \{0.1,0.25,0.5,1\}$ produce very similar social welfare. The latter is reasonable since, in the case of similar valuation functions, the allocation maximizing social welfare allocates the good roughly equally among the agents. The latter is a property that can be attained by both small and large values of $\alpha$. On the opposite side of the spectrum, in case the valuation function comes from different distributions, the optimal allocation needs to assign most of the items to a single agent. In such a case, we would expect smaller values of $\alpha$ leading to higher social welfare. The latter intuition is verified by our experimental evaluations presented in Table~\ref{tab:welfare2}. Notice that in this case, smaller values of $\alpha$ (e.g., $0.1$ or $0.25$) lead to significantly higher social welfare.

\begin{table*}[t]
\centering

\begin{minipage}{0.45\textwidth}
\centering
\small
\setlength{\tabcolsep}{4pt}

\begin{tabular}{|c|c||c|c|c|c|}
\hline
$\beta$ & $\alpha$ & $n=2$ & $n=3$ & $n=5$ & $n=10$ \\
\hline

\multirow{4}{*}{$1$}
& $0.1$  & 0.886 & 0.707 & 0.712 & 0.804 \\ \cline{2-6}
& $0.25$ & 0.887 & 0.707 & 0.724 & 0.810 \\ \cline{2-6}
& $0.5$  & \textbf{0.891} & \textbf{0.712} & 0.710 & 0.800 \\ \cline{2-6}
& $1$    & 0.886 & 0.711 & \textbf{0.737} & \textbf{0.814} \\
\hline

\multirow{4}{*}{$0.75$}
& $0.1$  & \textbf{0.574} & 0.737 & 1.008 & 1.272 \\ \cline{2-6}
& $0.25$ & 0.562 & \textbf{0.743} & 1.002 & 1.291 \\ \cline{2-6}
& $0.5$  & 0.562 & 0.740 & 1.000 & 1.300 \\ \cline{2-6}
& $1$    & 0.563 & 0.706 & \textbf{1.015} & \textbf{1.305} \\
\hline

\multirow{4}{*}{$0.5$}
& $0.1$  & 0.652 & 1.203 & 1.191 & 1.979 \\ \cline{2-6}
& $0.25$ & 0.645 & 1.194 & 1.189 & \textbf{1.988} \\ \cline{2-6}
& $0.5$  & 0.642 & 1.185 & \textbf{1.209} & 1.974 \\ \cline{2-6}
& $1$    & \textbf{0.654} & \textbf{1.214} & 1.167 & 1.982 \\
\hline

\end{tabular}

\caption{Each valuation function }$v_i(x_i) := \zeta_i \cdot x_i^\beta$ where $\zeta_i \in \mathrm{Unif}([0,1])$.
\label{tab:welfare1}
\end{minipage}
\hfill
\begin{minipage}{0.45\textwidth}
\centering
\small
\setlength{\tabcolsep}{4pt}

\begin{tabular}{|c|c||c|c|c|c|}
\hline
$\beta$ & $\alpha$ & $n=2$ & $n=3$ & $n=5$ & $n=10$ \\
\hline

\multirow{4}{*}{$1$}
& $0.1$  & 0.911 & \textbf{0.898} & \textbf{0.948} & \textbf{0.922} \\ \cline{2-6}
& $0.25$ & \textbf{0.956} & 0.898 & 0.901 & 0.877 \\ \cline{2-6}
& $0.5$  & 0.918 & 0.831 & 0.847 & 0.807 \\ \cline{2-6}
& $1$    & 0.829 & 0.718 & 0.639 & 0.504 \\
\hline

\multirow{4}{*}{$0.75$}
& $0.1$  & 0.900 & \textbf{0.960} & \textbf{0.911} & \textbf{0.919} \\ \cline{2-6}
& $0.25$ & \textbf{0.931} & 0.944 & 0.882 & 0.859 \\ \cline{2-6}
& $0.5$  & 0.904 & 0.865 & 0.833 & 0.815 \\ \cline{2-6}
& $1$    & 0.852 & 0.838 & 0.775 & 0.622 \\
\hline

\multirow{4}{*}{$0.5$}
& $0.1$  & 0.855 & \textbf{0.881} & \textbf{0.898} & \textbf{0.918} \\ \cline{2-6}
& $0.25$ & \textbf{0.898} & 0.864 & 0.878 & 0.887 \\ \cline{2-6}
& $0.5$  & 0.876 & 0.863 & 0.837 & 0.825 \\ \cline{2-6}
& $1$    & 0.849 & 0.778 & 0.815 & 0.742 \\
\hline

\end{tabular}

\caption{$v_i := \zeta_i \cdot x_i^\beta$ where $\zeta_1 \in \mathrm{Unif}([0.9,1])$ and $\zeta_i \in \mathrm{Unif}([0,0.1])$ for $i \neq 1$.}
\label{tab:welfare2}
\end{minipage}
\end{table*}

\section{Conclusion}
We introduce a unified framework for designing simple resource allocation mechanisms with proportional-style allocations and uniform pricing. Our framework yields the family of $\alpha$-proportional mechanisms, which interpolate between the Kelly mechanism and the first-price auction. We established the existence and uniqueness of Nash equilibria, proved improved efficiency guarantees over the Kelly mechanism, and showed that the Price of Anarchy approaches $1$ as $\alpha \rightarrow 0$. An interesting future research direction is in establishing Liquid Price of Anarchy guarantees in the presence of budget-constrained bidders, as well as considering more general polytope feasibility constraints.

\textbf{Aknowledgements} 
This work was partially supported by the Independent Research Fund Denmark (DFF) under grant 2032-00185B, the Villum Young Investigator Award, project number:~72091, and by project MIS 5154714 of the National Recovery and
Resilience Plan Greece 2.0 funded by the European Union under the NextGenerationEU Program.

\bibliography{refs.bib}
\bibliographystyle{plain}


\newpage\appendix

\section{Related Work}\label{app:related}
In Table~\ref{t:1}, we summarize the properties of the $\alpha$-proportional mechanism, the mechanism of Johari and Tsitsiklis~\cite{JT04} and the ESPA mechanism~\cite{MB06}.
\renewcommand{\arraystretch}{1.4}
\begin{table}[h]
\centering
\scriptsize
\setlength{\tabcolsep}{4pt}
\renewcommand{\arraystretch}{1.1}

\resizebox{\textwidth}{!}{%
\begin{tabular}{|l|c|c|c|c|c|c|}
\hline
\textbf{Name} 
& \textbf{Bidding Space}
& \textbf{Allocation} 
& \textbf{Pricing} 
& \textbf{Best NE Eff.} 
& \textbf{Worst NE Eff.} 
& \textbf{Unit Price} \\
\hline

VCG 
& $\tilde v_i : [0,1] \to \mathbb{R}$ 
& $\arg\max_{\mathbf{x}} \sum_i \tilde v_i(x_i)$ 
& externality payment 
& $1$ 
& $0$ 
& non-uniform \\
\hline

Scalar VCG 
& $\theta_i \in \mathbb{R}_{\ge 0}$ 
& $\arg\max_{\mathbf{x}} \sum_i \theta_i x_i^\alpha$ 
& externality payment 
& $1$ 
& $0$ 
& non-uniform \\
\hline

Kelly 
& $\theta_i \in \mathbb{R}_{\ge 0}$ 
& $\displaystyle x_i = \frac{\theta_i}{\sum_j \theta_j}$ 
& $\displaystyle p_i = \theta_i$ 
& $3/4$ 
& $3/4$ 
& uniform \\
\hline

$\alpha$-Prop. 
& $\theta_i \in \mathbb{R}_{\ge 0}$ 
& $\displaystyle x_i = \frac{\theta_i^{1/\alpha}}{\sum_j \theta_j^{1/\alpha}}$ 
& $\displaystyle p_i = \frac{\theta_i^{1/\alpha}}{\left(\sum_j \theta_j^{1/\alpha}\right)^{1-\alpha}}$ 
& $\frac{1 + 2\sqrt{\alpha}}{ (1 + \sqrt{a})^2}$ 
& $\frac{1 + 2\sqrt{\alpha}}{ (1 + \sqrt{a})^2}$
& uniform \\
\hline

ESPA 
& $\theta_i \in \mathbb{R}_{\ge 0}$ 
& $\displaystyle x_i = \frac{\theta_i}{\sum_j \theta_j}$ 
& $\begin{aligned}
p_i &= \left(\sum_{j\neq i} \theta_j\right) \\
&\quad \cdot \int_0^{\theta_i}
\frac{g(t+\sum_{j\neq i}\theta_j)}
{(t+\sum_{j\neq i}\theta_j)^2}\,dt
\end{aligned}$ 
& $1$ 
& $1$ 
& non-uniform \\
\hline

\end{tabular}%
}
\caption{Comparison of resource allocation mechanisms.}
\label{t:1}
\end{table}

\begin{remark}
Apart from its simplicity, an important advantage of the $\alpha$-proportional mechanism with respect ESPA and Scalar VCG is that it uses a fixed unit price. More precisely, this means that $p_i(\theta)/x_i(\theta)$ remains exactly the same for each agent $i \in [n]$. Fixed unit price is a very crucial property for fairness and regulatory reasons and we remark that $\alpha$-proportional mechanism is the first mechanism able to achieve arbitrarily high efficiency with fixed unit price.
\end{remark}

\begin{remark}
In case of $g(p) = p$ then the payment of ESPA mechanism~\cite{MB06} becomes $p_i(\theta_i,\theta_{-i}) = (\sum_{j \neq i} 
\theta_j) \ln(1 + s_i / (\sum_{j \neq i} s_j))$. 
\end{remark}

\section{Omitted Proofs of Section~\ref{s:framework}}\label{app:framework}
\lemmatwo*
\begin{proof}
Our mechanism solves the following optimization problem that maximizes the social welfare with respect to the proxy functions.
\[
\begin{aligned}
\max \quad & \sum_{i = 1}^n  \theta_i\cdot  \frac{x_i^{1-\alpha}}{1-\alpha} \\
\text{s.t.} \quad & \sum_i x_i \leq 1 \\
                  & x_i \ge 0 \quad \forall i \in [n]
\end{aligned}
\]
By taking the Lagrangian of the problem, we get
\[
L = \sum_i \theta_i \cdot \frac{x_i^{1-\alpha}}{1-\alpha} 
- \lambda \left(\sum_i x_i - 1\right) + \sum_{i=1}^n x_i \cdot k_i
\]
The value of the dual variable $\lambda$ will be the unit price of the good, which means that the payment of each agent equals $\lambda \cdot x_i$. By taking the partial derivatives with respect to each $i \in [n]$, we get
\[
\frac{\partial L}{\partial x_i} = \theta_i x_i^{-\alpha} - \lambda = 0
\]
\noindent which, in turn, implies
\[
x_i = \left(\frac{\theta_i}{\lambda}\right)^{\frac{1}{\alpha}} \text{ and } \lambda = \left(\sum_i \theta_i^{\frac{1}{\alpha}}\right)^\alpha
\]
So, the mechanism above can be equivalently described as each agent reporting a parameter $\theta_i \geq 0$ and then the allocation $x_i$ of each agent $ i \in [n]$ equals
\[
x_i = \frac{\theta_i^{\frac{1}{\alpha}}}{\sum_j \theta_j^{\frac{1}{\alpha}}} \text{ for some } \alpha \geq 1
\]
and the utility of each agent $i \in [n]$ equals,
\[u_i(\theta_i,\theta_{-i}) := v_i\left(\frac{\theta_i^{\frac{1}{\alpha}}}{\sum_j \theta_j^{\frac{1}{\alpha}}} \right) - \lambda \cdot \frac{\theta_i^{\frac{1}{\alpha}}}{\sum_j \theta_j^{\frac{1}{\alpha}}} = v_i\left(\frac{\theta_i^{\frac{1}{\alpha}}}{\sum_j \theta_j^{\frac{1}{\alpha}}} \right) -  \frac{\theta_i^{\frac{1}{\alpha}}}{\left(\sum_j \theta_j^{\frac{1}{\alpha}}\right)^{1-\alpha}}.\]
\end{proof}

\lemmaone*
\begin{proof}[Proof of Lemma~\ref{l:1}]
Consider the utility of agent $i \in [n]$ with respect to the proxy function $\tilde{v}_i(x_i)$. In this case, the utility is
\[\theta_i\cdot  \frac{x_i^{1-\alpha}}{1-\alpha} - \lambda \cdot x_i.\]
Notice that $\lambda = \theta_i \cdot x_i^{-\alpha}$ (see proof of Lemma~\ref{l:2} above). As a result, 
\[x_i := \mathrm{argmax}_{x_i \geq 0}~\left\{ \theta_i \cdot \frac{x_i^{1 - \alpha}}{1 - \alpha} - \lambda \cdot x_i\right\}.\]
\end{proof}

\section{Omitted Proofs of Section~\ref{s:NE}}\label{app:nash}

\lemmathree*
\begin{proof}
Let a strategy profile $\textbf{s} = (s_1,\ldots,s_n)$ and consider $S := \sum_{i=1}^n s_i$. Then, the $\alpha$-proportional mechanism 
allocates to each agent $i \in [n]$
\[
x_i(\s) =
\begin{cases}
s_i / S & \text{if } S > 0,\\
0 & \text{if } S = 0,
\end{cases}
\]
and charges payment
\[
p_i(\s) =
\begin{cases}
s_i / S^{1-\alpha}& \text{if } S > 0,\\
0 & \text{if } S = 0.
\end{cases}
\]
By the definition of $x_i(\s)$ and $p_i(\s)$, their gradients with 
respect to $s_i$ are as follows:
\begin{equation}\label{eq:a}
\frac{\partial x_i}{\partial s_i} := \frac{\sum_{j\neq i} s_j}{ S^2} = \frac{S - s_i}{S^2} = \frac{1 - x_i}{S}
\end{equation}

\begin{equation}\label{eq:b}
\frac{\partial p_i}{\partial s_i} := (\alpha - 1)S^{\alpha - 2}s_i + S^{\alpha 
- 1} = S^{\alpha - 1} \cdot( 1 - (1 - \alpha)x_i)
\end{equation}

Using Equations~\ref{eq:a} and ~\ref{eq:b}, we get the following corollary for the supergradients of $u_i(s_i,s_{-i})$. 

\begin{corollary}\label{l:chain}
Let $g_i \in \partial v_i(x_i)$ be a superdifferential of $v_i(\cdot)$ at $x_i$. Then 
\[ g_i \cdot \frac{1 - x_i}{S} - S^{\alpha - 1}\bigl(1 - (1-\alpha)x_i\bigr)\]
is a superdifferential of the function $u_i(\cdot,\s_{-i})$ at point $s_i$. 
\end{corollary}
\begin{proof}
Since  $u_i(s_i,s_{-i}) := v_i\left(s_i/S\right) - s_i/\left(S \right)^{1-\alpha}$ by the chain rule we get that for all $g_i \in \partial v_i(x_i^\star)$,
\[ g_i \cdot \frac{\partial x_i}{\partial s_i} - \frac{\partial p_i}{\partial s_i} \in \partial u_i(s_i,\s_{-i}).\]
Then by Equations~\ref{eq:a} and~\ref{eq:b} we get 

\[ g_i \cdot \frac{1 - x^\star_i}{S} - S^{\alpha - 1}\bigl(1 - (1-\alpha)x^\star_i\bigr) \in \partial_{s_i} u_i(s^\star_i,\s^{\star}_{-i}).\]
\end{proof}
\medskip
\medskip

Let $\x^\star \in \Delta_n$ be a maximizer of $\Psi(\x):=
\sum_{i=1}^n\int_0^{x_i} v'_{i+}(z) \cdot \frac{1 - z}{1 - (1-\alpha)z}dz$. Since $\Psi(\x)$ is differentiable and concave by the KKT conditions, we have that there exist $\mu \geq 0$ and $\lambda_i \geq 0$ for each $ i \in [n]$ such that
\[
0 =  \frac{\partial \Psi(\mathbf{x})}{\partial x_i} - \mu + \lambda_i = 0~~~~~~\text{for each } i \in [n].\]
Then, by Corollary~\ref{l:chain} we get that, for all $i \in [n]$,
\[
0 =  v'_{i+}(x^\star_i) \cdot \frac{1 - x^\star_i}{1 - (1-\alpha)x^\star_i} - \mu + \lambda_i.\]
\medskip

By the complementary slackness conditions, we have that for each $i \in [n]$,
\[
\lambda_i \cdot x^\star_i = 0.
\]
As a result, in case $x^\star_i > 0$ then $\lambda_i = 0$ so
\begin{equation}\label{eq:13}
\mu = v'_{i+}(x^\star_i) \cdot \frac{1 - x^\star_i}{1 - (1-\alpha)x^\star_i}.
\end{equation}
If $x^\star_i = 0$ then $v'_{i+}(0) - \mu + \lambda_i = 0$ for some $\lambda_i \ge 0$, which is equivalent to $
\mu \ge \sup \partial v_i(0).$

The latter completes the proof of Item~$1$ of Theorem~\ref{t:main}.

\medskip

Next, we establish Item~$2$ in Theorem~\ref{t:main}. We consider $S = \mu^{1/\alpha}$ and we will establish that the strategy profile defined as
\[ s^\star_i := S \cdot x^\star_i\]
forms a Nash Equilibrium. We will establish the latter claim by separately considering the cases of $x^\star_i > 0$ and $x_i^\star = 0$.

\noindent\emph{Case 1: $x^\star_i > 0$.} In this case, we have $s^\star_i >0$. We will establish that $s^\star_i \in \mathrm{argmax}_{s_i} u_i(s_i,\s^\star_{-i})$.

By Equation~\ref{eq:13} we get that,
\[
\mu = v'_{i+}(x^\star_i) \cdot  \frac{1 - x^\star_i}{1 - (1-\alpha)x^\star_i}.
\]
and by rearranging and setting $S := \mu^{1/\alpha}$ we get  
\begin{equation}\label{eq:11}
v'_{i+}(x^\star_i) \cdot \frac{1 - x^\star_i}{S}
= S^{\alpha - 1}\bigl(1 - (1-\alpha)x^\star_i\bigr).
\end{equation}

Together with Lemma~\ref{l:chain}, Equation~\ref{eq:11} ensures that 
$0 \in \partial_{s_i} u_i(\s^\star)$. In case the functions $u_i(s_i,\s^\star_{-i})$ were concave with respect to $s_i \geq 0$. The latter would directly imply that $s_i^\star := \mathrm{argmax}_{s_i} u_i(s_i,\s^\star_{-i})$.

Despite the fact that the function $u_i(s_i,\s_{-i})$ is not concave with respect to $s_i$, we establish that for any $s_i \neq s_i^\star$,
\begin{equation}\label{eq:not_zero}
   0 \notin \partial u_i(s_i,\s^\star_{-i}) 
\end{equation}
and thus $s^\star_i := \mathrm{argmax}_{s_i} u_i(s_i,\s^\star_{-i})$.
\medskip

We complete this part of the proof by establishing Equation~\ref{eq:not_zero}. Notice that $s^\star_i = x^\star_i \cdot S$ and $S := \sum_{j=1}^n s_j$. As a result, Equation~\ref{eq:11} can be written equivalently as 
\[
\frac{1}{S^{2}}\left( v'_{i+}(x^\star_i) \cdot \sum_{j \neq i} s^\star_j - S^{\alpha}\left(\sum_{j \neq i} s^\star_j + \alpha \cdot s^\star_i \right) 
\right) = 0.
\]
The latter in turn implies that
\begin{equation}\label{eq:12}
v'_{i+}(x^\star_i) \cdot \sum_{j \neq i} s^\star_j - S^{\alpha}\left(\sum_{j \neq i} s^\star_j + \alpha \cdot s^\star_i \right) = 0. \end{equation}
To this end, let us assume that $s'_i := \mathrm{argmax}_{s_i} u_i(s_i,\s^\star_{-i})$ and $s'_i \neq s^\star_i$. To simplify notation, let
\[ S' = s'_i + \sum_{j \neq i}s^\star_j ~~~~~\text{ and }~~~~~ x'_i = \frac{s'_i}{S'}.\]
Let $s'_i$ such that $0 \in \partial_{s_i} u_i(s'_i,\s^\star_{-i})$. Then corollary~\ref{l:chain} implies that there exists $g'_i \in \partial v_i(x'_i)$ such that
\[
g'_i \cdot \frac{1 - x'_i}{S'}
- S^{'\alpha - 1}\bigl(1 - (1-\alpha)x'_i\bigr) = 0.
\]
which in turn implies that 
\begin{equation}\label{eq:asd}
g'_i \cdot \sum_{j \neq i} s^\star_j - S^{'a}\left(\sum_{j \neq i} s_j + \alpha \cdot s'_i \right) = 0.    
\end{equation}

We will show that Equation~\ref{eq:asd} leads to a contradiction. Let $s'_i > s^\star_i$ then $x'_i > x^\star_i$ and $S' > S$. Since $v_{i}(x_i)$ is a concave function, the fact that $x'_i > x^\star_i$
implies that $g'_i \leq v'_{i+}(x^\star_i)$. Thus,
\[g'_i \cdot \sum_{j \neq i} s^\star_j - S^{'a}\left(\sum_{j \neq i} s^\star_j + \alpha \cdot s'_i \right) < v'_{i+}(x^\star_i) \cdot \sum_{j \neq i} s^\star_j - S^{\alpha}\left(\sum_{j \neq i} s^\star_j + \alpha \cdot s^\star_i \right) = 0.\]
The latter contradicts Equation~\ref{eq:asd}. Symmetrically for the case $s'_i < s_i $.

\medskip

\noindent\emph{Case 2: $x^\star_i = 0$.} In this case, $\mu \ge \sup \partial v_i(0)$ and thus $s_i^\star = 0$ is the maximizer of $u_i(s_i,\s^\star_{-i})$.
\end{proof}

\lemmafour*
\begin{proof}
Let $s^\star = (s_1^\star,\ldots, s_n^\star)$ a Nash Equilibrium and $S = \sum_{j=1}^n s^\star_j$. In case $x_i^\star >0$ we have that $\frac{\partial u_i(\s^\star)}{\partial s_i} = 0$ which by Corollary~\ref{l:chain} implies that
\[ v'_i(x_i^\star) \cdot \frac{1 - x^\star_i}{S} = S^{\alpha - 1}\bigl(1 - (1-\alpha)x^\star_i\bigr).\]
that in turn implies that
\[ v'_i(x_i^\star) \cdot \frac{1 - x^\star_i}{\bigl(1 - (1-\alpha)x^\star_i\bigr)} = S^{\alpha}.\]
Setting $\mu := S^\alpha$ we get the claim for $x^\star_i >0$. Now in case $x_i^\star = 0$ we have that
\[ v'_i(0) \cdot \frac{1}{S} < S^{\alpha - 1}\]
which implies the claim for $x_i^\star = 0$. Then the KKT conditions directly imply that $\x^\star \in \Delta_n$ maximizes $\Psi(\x)$. 
\end{proof}


\section{Omitted Proofs of Section~\ref{s:PoA}}

\theoremfive*
\begin{proof}
Consider a resource allocation instance with $n$ agents. Agent $1$ has valuation function $v_1(x)=x$; all other agents have the same valuation function $v_2(x)=\gamma\cdot x$ for 
$$\gamma=\frac{(n\, (1+\sqrt{\alpha})-1-2\sqrt{\alpha}+\alpha}{(n\, (1+\sqrt{\alpha})-1-2\sqrt{\alpha})\cdot (1+\sqrt{\alpha})}.$$

Notice that the strategy profile inducing the allocation $\x$ with $x_1=\frac{1}{1+\sqrt{\alpha}}$ and $x_i=\frac{\sqrt{\alpha}}{(n-1)\cdot (1+\sqrt{\alpha})}$ for $i=2, ..., n$ is a Nash equilibrium. Notice that for $\mu := 1/ (1 + \sqrt{\alpha})$
\begin{align*}
    \frac{1-x_1}{1-(1-\alpha)\cdot x_1} = \frac{1}{1+\sqrt{\alpha}} := \mu 
\end{align*}
and for any agent $i > 1$, 
\begin{align*}
\gamma\cdot \frac{1-x_i}{1-(1-\alpha)\cdot x_i} = \frac{1}{1+\sqrt{\alpha}} = :\mu
\end{align*}
This means that Item~$1$ of Theorem~\ref{t:main} is satisfied and thus $\s$ is a Nash Equilibrium.

The social welfare of $\x$ is 
\begin{align*}
    \frac{1}{1+\sqrt{\alpha}}+\gamma\cdot \frac{\sqrt{\alpha}}{1+\sqrt{\alpha}},
\end{align*}
which approaches $\frac{1+2\sqrt{\alpha}}{(1+\sqrt{\alpha})^2}$ from above as $n$ approaches infinity. Also, note that the definition of $\gamma$ guarantee that $\gamma<1$ for all $\alpha\in (0,1)$ and $n\geq 2$. Hence, the optimal social welfare is $1$ (giving the whole resource to player $1$), and the price of anarchy approaches $\frac{(1+\sqrt{\alpha})^2}{1+2\sqrt{\alpha}}$ from below as $n$ approaches infinity.
\end{proof}

\section{Omitted Proofs of Section~\ref{s:rev}}
\theoremsix*	
	\begin{proof}
	\textit{VCG Revenue for Linear Valuations.}
	Assume the players are ordered such that $\zeta_1 \ge \zeta_2 \ge \dots \ge \zeta_n \geq 0$. In this case, the VCG mechanism acts as a second-price auction. It produces the optimal allocation $\vec{x}^*$, where $x_1^* = 1$ and $x_i^* = 0$ for all $i > 1$. The only player with positive payment is player $1$ whose payment is $\zeta_2$. Hence, $\RVCG = \zeta_2$. 
	
	\smallskip\noindent\textit{Revenue of $\alpha$-Proportional in Equilibrium.}
	The revenue of the $\alpha$-proportional mechanism is the equilibrium price $\lambda$. For linear valuations, each player's subgradient at the equilibrium allocation $\vec{x}$ is the constant coefficient, i.e., we have that $g_i = \zeta_i$ for all players $i$. The equilibrium condition for any player $i$ is:
	\[ \zeta_i = \lambda \left( 1 + \alpha \frac{x_i}{1-x_i} \right) \]
	Applying this to the second player ($i=2$), we conclude that: 
	\[
		\zeta_2  =  \lambda \left( 1 + \alpha \frac{x_2}{1-x_2} \right) 
		       \leq \lambda (1+\alpha)\,,
	\]
	where the inequality follows from the fact that the second player's allocation is $x_2 \leq \frac{1}{2}$, because $x_2 \leq x_1$, since $\zeta_1 \geq \zeta_2$, and $ x_1+x_2 \le 1$. Therefore, $\frac{x_2}{1-x_2} \leq 1$. 
	\end{proof}

	\begin{theorem}\label{thm:vcg_bound}
	Let each player $i$ have a concave valuation $v_i: [0,1] \to \mathbb{R}_{\geq 0}$, let $\vec{x} = (x_1, \ldots, x_n)$ be the equilibrium allocation of $\alpha$-proportional and let $g_i \in \partial v_i(x_i)$ be a subgradient of player $i$'s valuation at equilibrium. Then, 
	\[ \RVCG \leq \sum_{i=1}^n \sum_{j \neq i} g_j \big(y_j^{-i} - x_j\big) 
	         \leq \sum_{i=1}^n g_{\max(-i)} - (n-1) \sum_{i=1}^n g_i x_i \,,
	\] 
	where $g_{\max(-i)} = \max_{j\neq i} \{ g_j \}$ and $\vec{y}^{-i}$ is the optimal allocation with player $i$ absent. 
	\end{theorem}

	\begin{proof}
	\textit{Bounding VCG Payments.}
	Let $\vec{x}^* = (x_1^*, \dots, x_n^*)$ be the optimal allocation produced by VCG. The VCG payment for each player $i$ is:
	\begin{equation}\label{eq:vcg_payment}
		\RVCG^i = \sum_{j \neq i} v_j(y_j^{-i}) - \sum_{j \neq i} v_j(x_j^*)
	= \sum_{j \neq i} \Big[ v_j(y_j^{-i}) - v_j(x_j) \Big] + \sum_{j \neq i} \Big[ v_j(x_j) - v_j(x_j^*)\Big]\,,
	\end{equation}
	where the second equality is obtained by adding and subtracting the valuation $\sum_{j \neq i} v_j(x_j)$ of the remaining players at equilibrium. 
		
	\smallskip\noindent\textit{Bounding the First Term by Concavity.}
	Due to concavity of valuation functions $v_j$, for any subgradient $g_j \in \partial v_j(x_j)$ and for any point $y$, $v_j(y) - v_j(x_j) \le g_j(y - x_j)$. Applying this to the first term of our payment equation (with $y = y_j^{-i}$ for each $j \neq i$), we obtain that: 
	\[ \sum_{j \neq i} \Big[ v_j(y_j^{-i}) - v_j(x_j) \Big] \le \sum_{j \neq i} g_j\big(y_j^{-i} - x_j\big) \]
	Substituting this back into \eqref{eq:vcg_payment} yields the following upper bound on the individual payment:
	\[ \RVCG^i \le \sum_{j \neq i} g_j\big(y_j^{-i} - x_j\big) + \sum_{j \neq i} \Big[ v_j(x_j) - v_j(x_j^*) \Big] \]
	
	Summing up the inequality above over all players $i = 1, \ldots, n$, we obtain that: 
	\[ \RVCG = \sum_{i=1}^n \RVCG^i \le \sum_{i=1}^n \sum_{j \neq i} g_j\big(y_j^{-i} - x_j\big) + \sum_{i=1}^n \sum_{j \neq i} \Big[ v_j(x_j) - v_j(x_j^*) \Big] \]
	
	\smallskip\noindent\textit{Bounding the Second Term by Optimality.}
	We next show that the second term of the above upper bound is always non-positive, due to the optimality of the allocation $\vec{x}^*$ produced by VCG. We note that each player $j$'s valuation term, $v_j(x_j) - v_j(x_j^*)$, appears exactly $n-1$ times in the second double sum. Therefore, we can rewrite the second double sum as:
	\[ \sum_{i=1}^n \sum_{j \neq i} \Big[ v_j(x_j) - v_j(x_j^*) \Big] = 
	   (n-1) \sum_{j=1}^n \Big[ v_j(x_j) - v_j(x_j^*) \Big] =
	   (n-1) \left( \sum_{j=1}^n v_j(x_j) - \sum_{j=1}^n v_j(x_j^*) \right)
	\]
	By definition, $\vec{x}^*$ is the optimal allocation. Hence, the above sum is always non-positive, which concludes the proof of the theorem's first inequality.
	
	\smallskip\noindent\textit{Proving the Second Inequality.}
	To show the first inequality of the theorem, we observe that 
	$$ \sum_{j \neq i} g_j y_j^{-i} \le g_{\max(-i)}\,,$$
	because $g_j \leq g_{\max(-i)} = \max_{j \neq i} \{ g_j \}$ and because $\vec{y}^{-i}$ is a feasible allocation, hence $\sum_{j\neq i} y^{-i}_j \leq 1$.
	Moreover, we observe that each term $g_j x_j$ appears exactly $n-1$ times in the double sum $\sum_{i=1}^n \sum_{j\neq i} g_j x_j$\,. Hence the double sum is simplified to $(n-1)\sum_{i=1}^n g_i x_i$\,. 
	\end{proof}
		
	\begin{corollary}\label{cor:revenue-general}
	Let $\vec{x} = (x_1, \ldots, x_n)$ be the equilibrium allocation of the $\alpha$-proportional mechanism, where the players are indexed so that $x_1 \geq x_2	\geq \cdots \geq x_n$. Then, 
	\[ \RVCG \leq \Rev\left[ 1+ \alpha \frac{x_2}{1-x_2} + \alpha(n-1) \sum_{i=1}^n\left(\frac{x_1}{1-x_1}-\frac{x_i}{1-x_i}\right)x_i
	\right] \]
	\end{corollary}
	
	\begin{proof}
	We recall that the revenue $\Rev$ of the $\alpha$-proportional mechanism is the equilibrium price $\lambda$ and that the equilibrium condition for any player $i$ is:
	\[ g_i = \lambda \left( 1 + \alpha \frac{x_i}{1-x_i} \right)\,, \]
	where $g_i \in \partial v_i(x_i)$ is a subgradient of player $i$ at equilibrium. 
	
	By the equilibrium conditions above, there are subgradients $g_i \in \partial v_i(x_i)$ so that $g_1 \geq g_2 \geq \cdots \geq g_n$. Then, the first sum in the second bound of Theorem~\ref{thm:vcg_bound} can be written as:
	\[ \sum_{i=1}^n g_{\max(-i)} = g_2+(n-1)g_1 = g_2+(n-1)\sum_{i=1}^n g_1 x_i\,,
	\]
	using $x_1 + \cdots+x_n = 1$. Therefore, the second bound of Theorem~\ref{thm:vcg_bound} can be written as:
	\begin{equation}\label{eq:vcg-revenue1}
		\RVCG \leq \sum_{i=1}^n g_{\max(-i)} - (n-1)\sum_{i=1}^n g_i x_i =
		g_2 + (n-1)\sum_{i=1}^n (g_1 - g_i)x_i\,.	
	\end{equation}
	Substituting the equilibrium conditions in \eqref{eq:vcg-revenue1} concludes the proof of the corollary. 
	\end{proof}

\theoremseven*
\begin{proof}
In the case of $n$ agents with an identical valuation function, the Nash Equilibrium is symmetric, meaning that each agent $i \in [n]$ admits $x_i = 1/n$. Then, the proof follows by applying the latter to Corollary~\ref{cor:revenue-general}.
\end{proof}

\theoremeight*
	\begin{proof}
    Let $\vec{x} = (x_1, x_2)$ be the equilibrium allocation of $\alpha$-proportional, with $x_1+x_2=1$. Applying the bound of Corollary~\ref{cor:revenue-general} for $n=2$ and using that $x_1+x_2 =1$, we obtain that:
	\begin{equation}\label{eq:two-players}
	\RVCG \leq \Rev \left[ 1+\alpha\left(\frac{x_2}{x_1} + \left(\frac{x_1}{x_2} - \frac{x_2}{x_1}\right) x_2 \right)\right]
	\end{equation}
    We observe that the multiplier of $\alpha$ in the parentheses in \eqref{eq:two-players} is: 
    \[ \frac{x_2}{x_1} + \left(\frac{x_1}{x_2} - \frac{x_2}{x_1}\right) x_2 = x_1 + (1-x_2)\frac{x_2}{x_1} = x_1 + x_2 = 1\,,
    \]
    which concludes the proof of the corollary. 
	\end{proof}

\end{document}